\documentclass[11pt]{article}
\usepackage[utf8]{inputenc}
\title{Federated Heavy Hitter Recovery under Linear Sketching}
\author{Adri\`a Gasc\'on \and Peter Kairouz\and Ziteng Sun \and Ananda Theertha Suresh}
\date{
\texttt{\{adriag, kairouz, zitengsun, theertha\}@google.com} \\[2ex]
Google Research}

\usepackage{lmodern}
\usepackage{xspace}
\usepackage{amsfonts,amsmath,amssymb,amsthm, pbox}

\usepackage[colorlinks,citecolor=blue,bookmarks=true]{hyperref}

\usepackage{multirow}
\usepackage{dsfont} %
\usepackage{fullpage}

\usepackage{algorithmic,algorithm}

\usepackage[shortlabels]{enumitem}
\setitemize{noitemsep,topsep=0pt,parsep=0pt,partopsep=0pt}
\setenumerate{noitemsep,topsep=0pt,parsep=0pt,partopsep=0pt}

\makeatletter
\@ifundefined{theorem}{%
  \theoremstyle{definition}
  
  \theoremstyle{plain}
  \newtheorem{theorem}{Theorem}
  
  \newtheorem{lemma}{Lemma}

  \theoremstyle{remark}
  \newtheorem{remark}{Remark}

}{}
\makeatother
\newenvironment{proofof}[1]{\noindent{\bf Proof of {#1}:~~}}{\(\qed\)}

\newcommand{\ignore}[1]{}

\newcommand{\EE}{\mathbb{E}}

\newcommand{\RR}{\mathbb{R}}

\newcommand{\expectation}[1]{\EE\left[#1\right]}

\def \cA     {{\cal A}}

\def \cN     {{\cal N}}

\def \cX     {{\cal X}}

\newcommand{\eg}{\textit{e.g.,}\xspace}
\newcommand{\ie}{\textit{i.e.,}\xspace}  %
\newcommand{\iid}{\textit{i.i.d.}} %

\def \ceil#1{{\lceil{#1}\rceil}}

\def \floor#1{{\lfloor{#1}\rfloor}}

\newcommand{\absv}[1]{\left|#1\right|}

\def \paren#1{{({#1})}}
\def \Paren#1{{\left({#1}\right)}}

\newcommand{\eqdef}{{:=}}

\newcommand{\probof}[1]{\Pr\Paren{#1}}

\def \th     {{\rm th }}

\def\ignore#1{}

\newcommand{\bi}{\begin{itemize}}
\newcommand{\ei}{\end{itemize}}

\def\orpro{\mathop{\mathchoice
   {\vee\kern-.49em\raise.7ex\hbox{$\cdot$}\kern.4em}
   {\vee\kern-.45em\raise.63ex\hbox{$\cdot$}\kern.2em}
   {\vee\kern-.4em\raise.3ex\hbox{$\cdot$}\kern.1em}
   {\vee\kern-.35em\raise2.2ex\hbox{$\cdot$}\kern.1em}}\limits}

\def\andpro{\mathop{\mathchoice
 {\wedge\kern-.46em\lower.69ex\hbox{$\cdot$}\kern.3em}
 {\wedge\kern-.46em\lower.58ex\hbox{$\cdot$}\kern.25em}
 {\wedge\kern-.38em\lower.5ex\hbox{$\cdot$}\kern.1em}
 {\wedge\kern-.3em\lower.5ex\hbox{$\cdot$}\kern.1em}}\limits}

\def\simge{\mathrel{%
   \rlap{\raise 0.511ex \hbox{$>$}}{\lower 0.511ex \hbox{$\sim$}}}}

\def\simle{\mathrel{
   \rlap{\raise 0.511ex \hbox{$<$}}{\lower 0.511ex \hbox{$\sim$}}}}

\newcommand{\idc}[1]{{\mathbbm{1} \left\{ #1\right\}}}

\newcommand{\norm}[1]{\lVert{#1}\rVert}

\renewcommand{\algorithmicrequire}{\textbf{Input:}}

\newcommand{\ns}{n}
\newcommand{\nspu}{m}
\newcommand{\dist}{\tau}

\newcommand{\ahist}{ApproxHist}
\newcommand{\ahh}{ApproxHH}
\usepackage{natbib}
\usepackage{graphicx}
\usepackage{fullpage}

\usepackage{bbm}
\renewcommand{\ignore}[1]{}
\usepackage[capitalise]{cleveref}

\def\withcolors{0}
\def\withnotes{0}

\ifnum\withnotes=1
\renewcommand{\th}[1]{ {\noindent \textit{\small\textcolor{blue}{theertha: #1}}}}
\newcommand{\zs}[1]{{\noindent \textit{\small\textcolor{red}{ziteng: #1}}}}
\newcommand{\znote}[1]{ [\textcolor{purple}{Ziteng: #1}] }
\newcommand{\tnote}[1]{ [\textcolor{blue}{Theertha: #1}] }
\newcommand{\anote}[1]{ [\textcolor{orange}{Adria: #1}] }
\newcommand{\todo}[1]{ [\textcolor{red}{TODO: #1}] }
\else
\renewcommand{\th}[1]{{}}
\newcommand{\zs}[1]{{}}
\newcommand{\znote}[1]{}
\newcommand{\tnote}[1]{}
\newcommand{\anote}[1]{}
\newcommand{\todo}[1]{}
\fi

\ifnum\withcolors=1
\newcommand{\new}[1]{{\color{red}{#1}}}

\else
\newcommand{\new}[1]{{{#1}}}

\fi
\newcommand{\linagg}{LinSketch}
\newcommand{\thr}{\tau}
\newcommand{\deciblt}{{\rm Dec}}
\newcommand{\rep}{b}
\newcommand{\heavy}{H}
\newcommand{\hash}{g}
\newcommand{\size}{C}
\newcommand{\width}{W}
\newcommand{\length}{H}
\newcommand{\mmax}{M_{\rm max}}
\newcommand{\lmax}{\ell_{\rm max}}
\newcommand{\indic}[1]{\mathbbm{1}\left\{ #1 \right\}}

\begin{document}

\maketitle

\begin{abstract}
Motivated by real-life deployments of multi-round federated analytics with secure aggregation, we investigate the fundamental communication-accuracy tradeoffs of the heavy hitter discovery and approximate (open-domain) histogram problems under a linear sketching constraint. We propose efficient algorithms based on local subsampling and invertible bloom look-up tables (IBLTs).  We also show that our algorithms are information-theoretically optimal for a broad class of interactive schemes. The results show that the linear sketching constraint does increase the communication cost for both tasks by introducing an extra linear dependence on the number of users in a round. Moreover, our results also establish a separation between the communication cost for heavy hitter discovery and approximate histogram in the multi-round setting. The dependence on the number of rounds $R$ is at most logarithmic for heavy hitter discovery whereas that of approximate histogram is $\Theta(\sqrt{R})$. We also empirically demonstrate our findings.
\end{abstract}

\section{Motivation}
\label{sec:motivation}

Collecting and aggregating user data drives improvements in the app and web ecosystems. For instance, learning popular out-of-dictionary words can improve the auto-complete feature in a smart keyboard, and discovering malicious URLs can enhance the security of a browser. However, sharing user data directly with a service provider introduces several privacy risks. 

It is thus desirable to only make aggregate data available to the service provider, rather than directly sharing (unanonymized) user data with them. This is typically achieved via multi-party cryptographic primitives, such as a \textit{secure vector summation} protocol \citep{MelisDC16, bonawitz2017secure, TelemtryPrio, BellBGL020}. For instance, for closed domain histogram applications, the users can first “one-hot” encode their data into a vector of length $d$ (the size of the domain) and then participate in a secure vector summation protocol to make the aggregate histogram (but never the individual user contributions) available to the service provider. 

\paragraph{Federated heavy hitters recovery. } The abovementioned solution requires $\Omega(d)$ communication. However, in many real life applications the domain size is very large or even unknown a priori. For example, the set of new URLs can be represented via 8-bit character strings of length 100, and can thus take $d = 256^{100}$ values, which is clearly impossible to support in practice. In such settings, linear\footnote{Linearity is necessary because non-linear compression/sketching schemes would not work under the secure vector summation primitive which only makes the sum of client-held vectors available to the server.} sketching is often used to reduce the communication load.
For example, ~\citet{MelisDC16} use secure count-min sketch aggregation 
for privacy preserving training of recommender systems, and ~\citet{Corrigan-GibbsB17} rely on count-min sketches
for gathering browser statistics, i.e. aggregate histogram queries.
Similarly, ~\citet{Hu0LGWGLD21} rely on secure aggregation of variants of Flajolet-Martin sketches for distributed cardinality estimation. \citet{bnoeh2021lightweight} uses sketching to reduce the cost for distributed subset-histogram queries. In the work closest to ours, \citet{chen2022secagg} show that
count-sketches can be used to recover the \textit{heavy hitter} items (i.e. frequently appearing items) while reducing the communication overhead. All these protocols operate in the single-round setting.

\paragraph{Sketching in multi-round aggregation schemes.} Even though count-sketches are great step towards solving the heavy hitters problem, this approach has only been analyzed in the single round data aggregation setting. However, most commonly deployed systems for federated analytics employ \textit{multi-round} schemes for data aggregation \citep{fl-sys}. This is primarily because (a) not all users are available around the same time, (b) the population may be very large (in the billions of devices) and therefore the server has to aggregate data over batches for bandwidth/compute reasons, and (c) running the cryptographic secure vector summation protocol has compute and communication costs that are super linear in the number of users we are aggregating over \citep{BellBGL020, bonawitz2017secure}. Further, count sketch based approaches have a decoding runtime that is linear in $d$, which is infeasible in the open domain setting, and improving it to $\log d$ involves blowing up the communication cost by the same factor. 

\paragraph{Our contributions.}
Our paper thus takes a principled approach towards uncovering the fundamental accuracy-communication tradeoffs of the heavy hitters recovery problem under the linearity constraints imposed by secure vector summation protocols. \new{We show that linearity constraints increase the per-user communication complexity. For a fixed total number of users, as the number of rounds increases, the required communication decreases due to less stringent linearity constraint.} 
\new{Moreover, surprisingly, }we show that count-sketches are strictly sub-optimal for this application, and we develop a novel provably optimal approach that combines client-side (local) subsampling with inverse Bloom lookup tables (IBLTs). Roughly speaking, we show (via lower bounds) that in the $R$-round case, any approach that solves an approximate histogram problem (with additive error)  will incur a $\sqrt{R}$ factor penalty in the communication cost, while our optimal approach incurs $\log(R)$.
Hence, even non-trivial modifications of count-sketches \new{and other frequency oracle-based algorithms} are strictly sub-optimal. 

We also empirically evaluate our proposed algorithms and compare it with count-sketch baselines. Significant advantage of our algorithm is observed, especially when $R$ is large. In the setting of \cref{fig:communication_m10000}, to achieve an F1 score of 0.8, we see a 10x improvement in communication compared to the baseline using Count-sketch.

\paragraph{Organization.} We formally define the problem in \cref{sec:prelim} and then discuss our results in \cref{sec:results}. Algorithms for heavy hitter recovery and approximate histogram are presented in \cref{sec:ahh} and \cref{sec:ahist}, respectively. We discuss a practical modification of our algorithm in \cref{sec:adaptive} and present the experimental results in \cref{sec:exp}.

\section{Problem setup and preliminaries}
\label{sec:prelim}
We consider heavy hitter discovery in the distributed setting with multi-round communication between the users and a central server. 
Suppose users come in $R$ rounds. In round $r \in [R]$, there are $\ns$ users, denoted by the set $B_r$. %
We assume the sets are pairwise disjoint, \ie $\forall r \neq r', B_r \cap B_r' = \emptyset$. Each user $i \in B_r$ contributes $m_i$ samples with a contribution bound $m_i \le m$ from a finite domain $\cX$ of size $d$. Let $h_i$ denote the user's local histogram where $\forall x \in \cX$, $h_i(x)$ is the number of times element $x$ appeared in user $i$'s local samples. By assumption, we have $\norm{h_i}_1 = m_i \le m$. Let $h^{(r)}$ be the aggregated histogram in round $r$, \ie
\[
\forall x \in \cX: \quad \quad h^{(r)}(x) = \sum_{i \in B_r} h_i(x).
\]
The aggregated histogram across all $R$ rounds is denoted by $h^{[R]}$ where
\[
    \forall x \in \cX: \quad \quad h^{[R]}(x) = \sum_{r \in [R]} h^{(r)}(x).
\]

The total number of users is denoted by $N \eqdef n R$. 
We will focus on cases where $d \gg Nm$, \ie the case where the support is large and the data is sparse.

The goal of the server is to learn useful information about the aggregated histogram $h^{[R]}$. More precisely, we consider the two tasks described below.

\ignore{
\paragraph{$(\thr, \dist)$-heavy hitter (\ahh).}  For a given threshold $\thr$, the goal of $\dist$-approximate heavy hitter recovery on the entire data stream is to return a set $\heavy$ such that with probability $1 - \beta$,
 \begin{enumerate}
     \item If $h^{[R]}(x) \ge \thr$, $x \in \heavy$.
     \item If $h^{[R]}(x) \le \thr - \dist$, $x \notin \heavy$.
 \end{enumerate}
 To simplify the notation, for most parts of the paper, we will consider heavy hitter recovery for $\dist = 9\thr/10$. And we will denote this problem as $\thr$-approximate heavy hitter. The results naturally extend to cases where $\dist = C \thr$ for some constant $C$ without changing the communication complexity up to a constant factor.
}

\paragraph{$\dist$-heavy hitter (\ahh).}  For a given threshold $\dist$, the goal of $\dist$-heavy hitter recovery on the entire data stream is to return a set $\heavy$ such that with probability $1 - \beta$,
 \begin{enumerate}
     \item If $h^{[R]}(x) \ge \dist$, $x \in \heavy$.
     \item If $h^{[R]}(x) \le \dist/10$, $x \notin \heavy$.
 \end{enumerate}
\paragraph{$\dist$-approximate histogram (\ahist).} The goal is to return an approximate histogram $\widehat{h}^{[R]}$ such that with probability $1 - \beta$,
\[
    \forall x \in \cX, \quad \quad \absv{\widehat{h}^{[R]}(x) - {h}^{[R]}(x)} \le \dist.
\]

It can be seen that $\dist/3$-approximate histogram is a harder problem than $\dist$-heavy hitter (HH) since an $\dist/3$-approximate histogram would imply a set of approximate heavy hitters by returning $H$ to be the list of elements with approximate frequency more than $\thr - \dist/2$. \new{Previous work often solves \ahh~by reducing it to \ahist~(\eg, in \citet{chen2022secagg}). However, as we show in this paper, our work establishes a seperation between the two tasks in terms of the communication complexity in the multi-round setting.} 

\paragraph{Efficient decoding.} Since $d \gg Nm$, we require
efficient encoding (run by users) and decoding (run by the server). More precisely, the encoding/decoding time should be polynomial in $N, m, R, \log d, \log(1/\beta)$ and other parameters.

\paragraph{Per-user communication complexity.} We focus on distributed settings where each user has limited uplink communication capacity. In particular, each user must compress their local histogram $h_i$ to a message of bit length $\ell$, denoted by $Y_i$. And the server must solve the above tasks based on the received messages. The \emph{communication complexity} of each task is the \emph{smallest} bit length such that there exists a communication protocol to solve the task. 

\begin{table*}[ht]
\centering
\renewcommand{\arraystretch}{2}
\begin{tabular}{|c|c|c|c|}
\hline
 Task/Setting & 1-round  \linagg & $R$-round \linagg & Without \linagg~($R = N$)\\ \hline
  $\thr$-\ahh & $\tilde{\Theta} \Paren{\frac{mN}{\thr}}$& $\tilde{\Theta}\Paren{\frac{mN}{\thr R}}$ & $\tilde{\Theta}\Paren{\frac{m}{\thr}}$\\ \hline
  $\thr$-\ahist &  $\tilde{\Theta} \Paren{\frac{mN}{\thr}}$& $\tilde{\Theta} \Paren{ \min\{\frac{mN}{\thr\sqrt{R}}, \frac{mN}{R} \}}$ & $\tilde{\Theta}\Paren{\min\{ \frac{m \sqrt{N}}{\thr}, m\}}$ \\ \hline
\end{tabular}
\caption{Per-user communication complexity. All described bounds can be acheived by a \emph{non-interactive} protocol with server runtime $\text{poly}(m,n,R,\log(d), \log(1/\beta))$. Recall $N = nR$ denotes the total number of users. All bounds cannot be improved up to logarithmic factors
even under \textit{interactive} protocols. $\tilde{\Theta}$ omits factor that are logarithmic in $\tau, R, m$ and $N$.
}
\label{tab:results}
\end{table*}

\paragraph{Distributed estimation with linear sketching (\linagg).} A even more stringent communication model is the linear summation model. In each round $r$, each user $i \in B_r$ can only send a message $Y_i$ from a finite ring $G_r$ %
based on their local histogram and shared randomness $U$. 
For all $i \in B_r$, let
\[
    Y_i = f_i(h_i, U).
\]
Under the linear aggregation model, the server only observes 
\[
    Y^{(r)} = \sum_{i \in B_r} Y_i,
\]
where the addition is the additive operation in the ring $G_r$ and by definition, $Y^{(r)} \in G_r$.
The reason why we restrict ourselves
to a finite ring is for compatibility with cryptographic protocols for secure vector summation~\citep{bonawitz2017secure, BellBGL020}, which operates over a finite space.
These protocols ensure that any additional
information observed by the server beyond $Y^{(r)}$ can in fact be simulated given $Y^{(r)}$, under standard cryptographic assumptions. As mentioned above, we abstract away the specifics of the underlying protocol and assume that the server observes exactly $Y^{(r)}$.
For vector summation, it is convenient to think of $G_r$ as $\mathbb{Z}_{q_r}^\ell$, i.e. length-$\ell$ vectors with integer entries mod $q_r$ (we might chose $q_r$ to be prime when we require division, e.g. in the IBLT construction).

If the protocol is \emph{interactive}, for $i \in B_r$, $Y_i$ is allowed to depend on $Y^{(1)}, \ldots, Y^{(r-1)}$. In this case,
 each $f_i$ is a function of $Y^{(1)}, \ldots, Y^{(r-1)}$. If the protocol is \emph{non-interactive}, $f_i$'s must be fixed independently from previous messages.
 
The server then must recover heavy hitters (and their frequencies) based on the transcript of the protocol, denoted by
 \[
    \Pi = (Y^{(1)}, \ldots, Y^{(R)}, U).
 \]

\new{
\subsection{Connection to other constrained settings.}

Below we discuss the connection between our stated setting to other popular constrained settings including the streaming setting, and the general communication constrained setting. 

\textbf{Connection to the streaming setting.} When $R = 1$, the setting is similar to the streaming setting (\eg in \citet{CORMODE200558}) since they both require that all information about the dataset must be ``compressed'' into a small state. One important difference is that in the distributed setting, the local data is processed independently at each user and only linear operation on the state is permitted due to the linear aggregation operation.
For the $R$-round case, our setting is different since the server observes $R$ states, each with bit length at most $\ell$. These states couldn't be viewed as a mega-state with bit length $R\cdot \ell$ since there is a further restriction that each sub-state (corresponding to the aggregated message observed in a round) can only contain information about data in the corresponding round instead of the entire data stream. Another naive way to reduce the problem to the streaming setting is to sum over all $R$ states and obtain a single state with $\ell$ bits. However, our result implies that this reduction is strictly sub-optimal
(see \cref{tab:results}). To the best of our knowledge, similar settings have not been studied in either the federated analytics literature or the streaming literature. 

\textbf{Connection to distributed estimation without linear sketching.} The general communication constrained setting where linear sketching is not enforced could be viewed as a special case of the proposed framework with $n = 1$ and $R = N$ since in this case, the linear aggregation is performed only over one user's message and hence trivial.

We list comparisons to these settings for the considered tasks in \cref{tab:results}. %
}

\section{Results and technique}\label{sec:results}

\new{We consider both approximate heavy hitter recovery and approximate histogram estimation in the linear aggregation model. 
We establish tight (up to logarithmic factors) communication complexity for both tasks in the single-round and multi-round settings.
The results are summarized in \cref{tab:results}. Our results have the following interesting implications on the communication complexity of these problems.}

\paragraph{Linear aggregation increases the communication cost.} As shown in \cref{tab:results}, under \linagg, for both tasks, the per-user communication would incur a linear dependence on $\ns = N/R$, the number of users in each round. On the other hand, without linear aggregation constraint, there won't be a linear dependence on $\ns$ since each user can simply send their $\nspu$ local samples losslessly using $O(\nspu \log d)$ bits.
The result establishes the fundamental cost of linear aggregation communication protocols for heavy hitter recovery.

\paragraph{\ahist~is harder than \ahh.} As mentioned before, a natural way to obtain heavy hitters is to obtain an approximate histogram and do proper thresholding to select the heavy elements. Although in the single-round case, there is at most a logarithmic gap between the communication complexity for the two problems. In the $R$-round case, our result shows that this is strictly sub-optimal. More precisely, the communication cost for $\thr$-\ahh~increases by a factor of $\sqrt{R}$ while that of \ahist~depends at most logarithmically in $R$.
This implies a gap between the per-user communication cost for $\thr$-\ahist~and $\thr$-\ahh~in the multi-round case. 

\new{\paragraph{The impact of $R$.} With a fixed total number of users $N$, our result shows that the per-user communication complexity decreases as $R$ increases. This is due to the fact that as $R$ increases, the linearity constraints are imposed over a smaller group of users with size $N/R$, and hence less stringent. However, this also comes at the cost that the privacy implication from aggregation becomes weaker.}

\subsection{Our technique - IBLT with local subsampling}
As discussed above, when solving the approximate heavy hitter problem in the multi-round setting, algorithms that rely on obtaining an approximate histogram and thresholding won't give the optimal communication complexity. In the paper, we propose to use invertible bloom lookup tables (IBLTs) \citep{Goodrich2011iblt} and local subsampling. At a high-level, IBLT is a bloom filter-type linear data structure that supports efficient listing of the inserted elements and their exact counts. The size of the table scales linearly with the number of unique keys inserted. To reduce the communication cost, we perform local threshold sampling \citep{Duffield2005threshold} on users' local datasets. This guarantees that the ``light'' elements will be discarded with high probability and hence won't take up the capacity of the IBLT data structure. Compared to frequency-oracle based approach, the variance of the noise introduced in our subsampling-based approach for each item is proportional to its accumulative count, which gives better estimates for elements with frequencies near the threshold. For elements with counts way above the threshold, the frequency estimate will have a larger error but this won't affect heavy hitter recovery since only whether the count is above $\thr$ is crucial to our problem. 
See detail of the algorithm in \cref{sec:ahh_upper}.

\subsection{Related work}
Linear dimensionality reduction techniques for frequency estimation and heavy hitter recovery has been widely studied to reduce storage or communication cost, such as Count-sketch, Count-min sketch ~\citep{charikar2002finding, CORMODE200558, donoho2006compressed, minton2014improved}, and efficient decoding techniques have also been proposed \citep{cormode2006combinatorial, gilbert2010approximate}.

The closest to our work is \new{that} of \cite{chen2022secagg}, which studies approximate histogram estimation under linear sketching constraint. The work also establishes gap between communication complexities with/without Secure Aggregation. However, their result is in a more restricted setting of $m = 1$ and $R = 1$. Moreover, our algorithm also has the advantage of being computationally efficient (runtime only depends logartihmically in $d$), which is important for applications with large support but sparse data. \new{Their work also considers algorithms that guarantee distributed differential privacy guarantees, which we leave as interesting future directions.}

\section{Approximate heavy hitter under linear aggregation}
\label{sec:ahh}
In this section, we study the approximate heavy hitter problem and show that the problem can be solved with per-user communication complexity $\tilde{O}\Paren{\frac{mn}{\thr} \log d}$, stated in \cref{thm:ahh}.

A natural comparison to make is the heavy hitter recovery algorithm obtained from getting a frequency oracle up to accuracy $\Theta(\thr)$. Since there are $R$ rounds, the naive approach would require an accuracy of $\Theta(\thr/R)$ in each round and classic methods such as Count-min and Count-sketch would require a per-user communication complexity of $\tilde{\Theta}(\nspu\ns R/\thr)$. In the $R$-round case, our result improves upon this by a factor of $R$. In fact, as we show in \cref{thm:ahist_lower}, any frequency oracle-based approach would require per-user communication complexity of at least $\Omega(\nspu\ns \sqrt{R}/\thr)$. Our result improves upon these and show that the dependence on $R$ is at most logarithmic.

\begin{theorem} \label{thm:ahh}
There exists a non-interactive linear sketching protocol with communication cost %
$\tilde{O}\Paren{\frac{mn}{\thr} }$
bits per user, %
which solves the \textbf{$\thr$-approximate heavy hitter}
problem. Moreover, the running time of the algorithm is
$\tilde{O}\Paren{\frac{mn}{\thr}}$. 
\end{theorem}

The next theorem shows that the above communication complexity is minmax optimal up to logarithmic factors. 
\begin{theorem} \label{thm:ahh_lower}
    For any $\thr$ and \new{interactive linear sketching %
    protocol} $\cA$ with per-user communication cost $o\Paren{\frac{mn}{\thr}}$, there exists a dataset $h_i, i \in B_r, r \in [R]$, such that $\cA$ cannot solve $\thr$-heavy hitter (HH) with success probability at least 4/5.
\end{theorem}

Next we will present the protocol that achieves \cref{thm:ahh} in \cref{sec:ahh_upper} and discuss the proof of the lower bound \cref{thm:ahh_lower} in \cref{sec:ahh_lower}.

\label{sec:ahh_upper}

At a high level, the protocol relies on two main components: (i) a probabilistic data structure called Invertible Bloom Lookup Table (IBLT) introduced by \citet{Goodrich2011iblt}, and (ii) local subsampling. We start by introducing IBLTs, starting from the more standard (counting) Bloom filters. 

\paragraph{IBLT: Bloom filters with efficient listing.} Note that each user's local histogram $h_i$ can be viewed as a sequence of key-value pairs $(x, h_i(x))$.
The Bloom filter data structure
is a standard linear data structure to represent a 
set of key-value pairs with keys coming from a large domain. 
IBLT is a version of Bloom filter that %
supports an efficient listing operation -- while preserving the other nice properties of Bloom counting filters, namely linearity (and thus mergeable by summation), and succintness (linear size in number of indices it holds).\footnote{\new{In our algorithm, IBLT could be replaced by other data structures with these properties.}}
These properties are summarized in the following Lemma.

\begin{lemma}[\cite{Goodrich2011iblt}] \label{lem:iblt}
Consider a collection of local histograms $(h_i)_{i\in [n]}$ over $[d]$ such that $\norm{\sum_{i \in [\ns]} h_i}_0 \le L_0$.

For any $\gamma  > 0$, there exist local linear sketches $\{f_i\}_{i \in [\ns]}$ of length $\ell = \tilde{O}(\gamma L_0)$ and an $O(\ell)$ time decoding procedure $\deciblt(\cdot)$~such that 
\[
    \deciblt\paren{\sum_{i \in [\ns]} f_i(h_i)} = \sum_{i \in [\ns]} h_i
\]
 succeeds except with probability at most $O\Paren{L_0^{2 -\gamma}}$.
\end{lemma}

For the purpose of this paper we can focus on 
the two main operations supported by an IBLT instance $\mathcal{B}$ (see~\cite{Goodrich2011iblt} for details on deletions and look-ups):
\begin{itemize}
    \item $\texttt{Insert}(k, v)$, which inserts the pair $(k, v)$ into 
$\mathcal{B}$.
\item
$\texttt{ListEntries}()$,
which enumerates the set of 
key-value pairs in $\mathcal{B}$.
\end{itemize} Note that
$f_i(h_i)$ in Lemma~\ref{lem:iblt} corresponds to the IBLT $\mathcal{B}_i$ resulting from inserting the set $\{(x,h_i(x)) ~|~ h_i(x) > 0\}$ into an empty IBLT.
Also, $\texttt{ListEntries}()$
corresponds to $\deciblt$
in Lemma~\ref{lem:iblt}.

Finally, 
$\sum_i^n f_i(h_i)$ corresponds to 
the encoding of 
the IBLT resulting 
from inserting the set $\{(x, \sum_i^n h_i(x)) ~|~ \exists i\in[n]: h_i(x) > 0\}$ into an empty IBLT.
In other words, each client $i\in [n]$
computes {\em local}
IBLT $\mathcal{B}_i := f_i(h_i)$,
and the (secure) aggregation of the 
$\mathcal{B}_i$'s results in the
{\em global} IBLT 
$\mathcal{B}:= \sum_i^n f_i(h_i)$. Further details on IBLT are stated in \cref{sec:iblt_app}.

\paragraph{Reducing capacity via threshold sampling.}
The second tool in our main protocol is threshold sampling. Note that
the guarantee in \cref{lem:iblt} relies on the number of unique elements in $\sum_{i \in [\ns]} h_i$, which can be at most $\nspu \ns$ in the worst-case, leading to an $O(mn)$ worst-case communication cost, not matching our lower bound in Lemma~\ref{thm:ahh_lower}. For heavy hitter recovery, we reduce the communication cost by local subsampling. More precisely, we use the threshold sampling algorithm from \cite{duffield05}, detailed in \cref{alg:threshold_sampling}
to achieve the (optimal) dependency $O(mn/\tau)$. 

\new{
\begin{remark}
Threshold sampling can be replaced by any unbiased local subsampling method that offers sparsity, \eg binomial sampling where $p \cdot h'(x) \sim \text{Binomial}(h(x), p)$ for some $p \in (0,1)$, and similar theoretical guarantee will hold. In this work, we choose threshold sampling due to the property that it minimizes the total variance of $h'$ under an expected sparsity constraint (see \citet{duffield05} for details).
\end{remark}
}
\begin{algorithm}[h]
\caption{Threshold sampling.}
\begin{algorithmic}[1]
\STATE \textbf{Input:} $h:$ local histogram. $t \in \RR_{+}:$ threshold.

\FOR{$x \in {\rm supp}(h)$}
    \IF {$h(x) \ge t$, }
        \STATE $h'(x) = h(x)$.
    \ELSE
        \STATE
        \[
            h'(x) = \begin{cases}
                t & \text{ with prob } \frac{h(x)}{t}, \\
                0 & \text{ otherwise.}
            \end{cases}
        \]
    \ENDIF
\ENDFOR
\STATE \textbf{Return:} $h'$.
\end{algorithmic}
\label{alg:threshold_sampling}
\end{algorithm}

The protocol that achieves the desired communication complexity in Theorem~\ref{thm:ahh} is detailed in \cref{alg:subsample_IBLT}. 

\begin{algorithm}[h!]
\caption{Subsampled IBLT with \linagg.}
\begin{algorithmic}[1]
\STATE \textbf{Input:} $\{h_i\}_{i \in B_r, r \in [R]}:$ local histograms; $d:$ alphabet size; $R:$ number of rounds; $m:$ per-user contribution bound; $n:$ number of users per round; $\thr:$ threshold for heavy hitter recovery; $\beta:$ failure probability.
\STATE Let $t = \max\{\tau/2, 1\}$, $\rep = \ceil{10\log(\frac{4\nspu\ns R}{\tau\beta})}$ and $L_0 = 20 \frac{mn}{\tau} \log R, \gamma = \log R$.

\FOR{$r \in [R]$}
\FOR{$j \in [b]$}
\STATE Each user $i \in B_r$ applies \cref{alg:threshold_sampling} with threshold $t$ in to their local histogram with fresh randomness to get $h'_{i, j}$. \label{line:iblt_encoding}
\STATE Each user sends message
$
    Y_{i, j} = f_{i,j}(h'_{i, j})
$
where $f_{i, j}$'s are mappings from \cref{lem:iblt} with parameter $L_0, \gamma$ and fresh randomness.
\STATE Server observes $\sum_{i \in B_r} Y_{i,j}$ and computes $$\hat{h}_{r, j} = \deciblt (\sum_{i \in B_r} Y_{i,j}).$$
If the decoding is not successful, we set $\hat{h}_{r, j}$ be the all-zero vector.\label{line:iblt_decoding}
\ENDFOR
\ENDFOR
\FOR{$j \in [b]$}\STATE Server computes $\hat{h}^{[R]}_j = \sum_{r \in [R]} \hat{h}_{r,j},$
and obtain list 
\[
    \heavy_j = \{x \in [d] \mid  \hat{h}^{[R]}_j > 0\}.
\]
\ENDFOR
\STATE \textbf{Return: } 
\[
    \heavy = \{x \mid \sum_{j \in [\rep]} \idc{x \in \heavy_j} \ge \frac{\rep}{2} \}.
\]
\end{algorithmic}
\label{alg:subsample_IBLT}
\end{algorithm}

The algorithm can be viewed as $\rep \eqdef \ceil{20\log(\frac{40\nspu\ns R}{\tau\beta})}$ independent runs of a basic protocol, each of which returns a list $\heavy_i$ of potential heavy hitters. And the repetition is to boost the error probability.

In each basic protocol, users first apply \cref{alg:threshold_sampling} to subsample to the data, which reduces the number of unique elements while maintaining the heavy hitters upon aggregation. Then the user encodes their samples using IBLTs, whose aggregation is then sent to the server to decode. Since the number of unique elements is reduced through subsampling, the decoding of the aggregated IBLT will be successful with high probabiltiy, hence recovering the aggregation of subsampled local histograms. The detailed proof of \cref{thm:ahh} is presented in \cref{sec:proof_ahh}.

\section{Approximate histogram under linear aggregation}
\label{sec:ahist}
In this section, we study the task of obtaining an approximate histogram in the multi-round linear aggregation model. The first observation we make is that using \cref{alg:subsample_IBLT} with threshold $\dist$, we are able to return a list $H$ of heavy hitters such that with high probability, the list contains all $x$'s with frequency more than $\dist$ and no tail elements. The approximate histogram algorithm builds on this and further asks each user to send a linear sketching of the their unsampled local data alongside the IBLT data structures in \cref{alg:subsample_IBLT}. The server would then use the aggregation of these linear sketches as a frequency oracle to estimate the frequency of elements in $H$. 

The above protocol leads to near optimal performance in the single-round case.
However, the $R$-round case is trickier since the error will build up along all $R$ rounds and the naive application of the sketching algorithm will lead to an error that depends linearly in $R$. This can be solved by carefully designing the correlation among hash functions in all $R$ rounds and we show that the dependence on $R$ can be reduced to $\sqrt{R}$.
We further show that the $\sqrt{R}$ dependence is in fact optimal by proving a matching lower bound, stated in \cref{thm:ahist_lower}.

\new{To improve the dependence on $R$, we use the \textsc{HybridSketch} idea from \citet{wu2023private}}. More precisely, 
the location hashes are fixed across rounds while the sign hashes are generated with fresh randomness. The details of the algorithm are described in \cref{alg:ahist_r}. The proof follows from the guarantee in \cref{thm:ahh} and standard analysis for the Count-sketch algorithm. We defer the complete proof to \cref{sec:ahist_app}.  %

\begin{theorem}\label{thm:ahist_r}
      In the $R$-round setting, there exists a linear aggregation protocol with communication cost %
      \new{$\tilde{O}\Paren{
      \min\{ \frac{mn\sqrt{R}}{\dist}, mn\}}$}
      per user, which solves the \textbf{$\dist$-approximate histogram}
problem. Moreover, the running time of the algorithm is %
      \new{$\tilde{O}\Paren{
      \min\{ \frac{mn\sqrt{R}}{\dist}, mn\}}$}. 
\end{theorem}

\begin{algorithm}[h]
\caption{$R$-round \ahist~with \linagg}
\begin{algorithmic}[1]
\STATE \textbf{Input:} $\{h_i\}_{i \in B_r, r \in [R]}:$ local histograms; $d:$ alphabet size; $R:$ number of rounds; $m:$ per-user contribution bound; $n:$ number of users per round; $\dist:$ error for approximate histogram; $\beta:$ failure probability.
\new{\IF{$\tau \le \sqrt{R}$}
\STATE Users implement \cref{alg:subsample_IBLT} with $\tau = 1$ and \textbf{return} the histogram obtained in Line 11.
\ENDIF}
\STATE Let $w = \ceil{ \frac{10\nspu \ns\sqrt{R}}{\dist}}$ and $\rep = \ceil{\log\Paren{\frac{4\nspu\ns R}{\tau \beta}}}$.

\STATE Get the same set of location hash functions $\{g_j: [d] \rightarrow [w]\}_{j \in [w]}$ for all rounds. And the independent sets of sign hashes $\{s_{j,r}: [d] \rightarrow \{\pm 1\}\}_{j \in [w], r \in [R]}$ across rounds.
\FOR{$r \in [R]$}
\STATE (\emph{In Parallel}) Each user $i \in B_r$ implements the protocol in \cref{alg:subsample_IBLT} and sends messages $Y_i$.

\STATE (\emph{In Parallel}) User $i \in B_r$ encode  $j \in [b]$ and $k \in [w]$,  
\[
    T_{i}(j,k) =\sum_{x}\indic{ \hash_{j}(x) = k} s_{j, r}(x) \cdot h_i(x).
\]
\ENDFOR
\STATE Server obtains a list $H$ of heavy hitters from the the messages $Y_i$'s.

\STATE Server obtains $\forall r \in [R], T^{(r)} = \sum_{i \in B_r} T_i$ and constructs $\hat{h}$, where $\forall x \in H$
\[
   \hat{h}(x) = {\rm Median}\Paren{ \{\sum_{r \in [R]} T^{(r)}(j, \hash_{j}(x)) \cdot s_{j, r}(x)  \}_{j \in [\rep]}},
\]
and $\forall x \notin H, \hat{h}(x) = 0$.
\STATE \textbf{Return} $\hat{h}.$
\end{algorithmic}
\label{alg:ahist_r}
\end{algorithm}

\paragraph{Lower bound for \ahist} %
We prove the following lower bound on \ahist, which shows that the bound in \cref{thm:ahist_r} is tight up to logarithmic factors, establishing the seperation between the sample complexities from \ahh~and \ahist.

\begin{theorem}
\label{thm:ahist_lower}
      For any $\dist$ and $R$-round \ahist~protocol with per-user communication cost \new{$o\Paren{\min\{\frac{mn\sqrt{R}}{\dist}, mn\}}$}, there exists a dataset $\{ h_i\}_{i \in B_r, r \in [R]}$, such that the protocol cannot solve $\dist$-\textbf{approximate histogram} with error probability at most 1/5.
\end{theorem}

\section{Practical adaptive tuning for instance-specific bounds} \label{sec:adaptive}
In practical scenarios, the per-user communication cost $\ell$ is often determined by system constraints (\eg delay tolerance, bandwidth constraint) and the goal is to recovery heavy hitters with the small enough $\thr$ under a fixed communication cost $\lmax$. While we have shown in \cref{thm:ahh_lower}, in the worst case, we can only reliably recover heavy hitters with frequency at least $\Omega(\frac{\nspu\ns}{\lmax})$. However, since the successful decoding of IBLTs only requires the number of \emph{unique} elements in a round to be small, when users' data is more favorable, it is possible to obtain better instance-specific bounds when the data is more concentrated on ``heavy'' elements.

When interactivity across rounds is allowed, we give an adaptive tuning algorithm for the subsampling parameter, which can be implemented when interactivity is allowed. The details of the algorithm are described in \cref{alg:adaptive_iblt}. At a high level, our algorithm is based on an estimate for $\norm{\sum_{i \in B_r} h_i'}_0$ where $h_i'$s are the subsampled histograms. When the decoding is successful, we can compute $\norm{\sum_{i \in B_r} h_i'}_0$ exactly from the recovered histogram. When the decoding is not successful, we rely on an analysis based on the ``core size'' of a random hypergraph~\citep{molloy2005cores} introduced by the hashing process to get an estimate of $\norm{\sum_{i \in B_r} h_i'}_0$. We discuss this in details in \cref{sec:iblt_app}.
Under the assumption that for a fixed subsampling parameter $t$, $\norm{\sum_{i \in B_r} h_i'}_0$  will be relatively stable across rounds, we can then increase/decrease $t$ based on past estimates of the data process.

We will empirically demonstrate the effectiveness of our tuning procedure. We leave proving rigorous guarantees on the adaptive tuning algorithm as an interesting future direction.

\begin{algorithm}[h]
\caption{Adaptive subsampled IBLT}\label{alg:adaptive_iblt}
\begin{algorithmic}[1]
\REQUIRE{Communication budget $C$, number of users $\ns$, user contribution bound $\nspu$. \\
\textbf{Update}: A tuning function that updates the subsampling parameter based on past observations.}
\STATE Set $
    t_0 = \Theta\Paren{\frac{\ns m}{C}}.
$
\FOR{$r = 0 , 1, 2, \ldots, R$}
        \STATE Each user $i \in B_r$ applies \cref{alg:threshold_sampling} with threshold $t$ in to their local histogram with fresh randomness to get $h'_{i}$. \label{line:iblt_encoding}
\STATE Each user sends message
$
    Y_{i} = f_{i}(h'_{i})
$
where $f_{i}$'s are mappings from \cref{lem:iblt} with parameter $L_0, \gamma$ and fresh randomness.
\STATE Server observes $\sum_{i \in B_r} Y_{i}$ and computes $$\hat{h}_{r} = \deciblt \paren{\sum_{i \in B_r} Y_{i}}$$
If the decoding is not successful, we let $\hat{h}_{r, j}$ be the all-zero vector.\label{line:iblt_decoding} 
\IF{The decoding is successful,}
    \STATE Set $\hat{s}_r = \norm{\hat{h}_{r}}_0$.
    \ELSE
        \STATE Get an estimate $\hat{s}_r$ for $\norm{\sum_{i \in B_r} h_i'}_0$ based on $\sum_{i \in B_r} Y_{i}$ using \eqref{eq:cardinality-from-core} and \eqref{eq:cardinality-from-core-2} (\cref{sec:iblt_app}). 
\ENDIF
        \STATE Set
        \[
        t_{r + 1} = \textbf{Update}(t_r, C, \hat{s}_r).
        \]
\ENDFOR
\end{algorithmic}
\end{algorithm}

\section{Experiments} \label{sec:exp}
In this section, we empirically evaluate our proposed algorithms (\cref{alg:subsample_IBLT,alg:adaptive_iblt}) for the task of heavy hitter recovery and compare it with baseline methods including (1) Count-sketch based method; (2) IBLT-based method without subsampling (\cref{alg:subsample_IBLT} with $\tau = 1$). We measure communication cost 
in units of words (denoted as $\size$) and each word unit is an {\rm int16} object (can be communicated with 2 bytes) in python and $C$++ for implementation purposes. 
We will mainly focus on string data with characters from $\rm ROOT$ consisting of lower-case letters, digits and special symbols in $\{'~@~\#~-~;~*~:~.~/~\_\}$. %
Below are the details of our implementation. 

\paragraph{Count-median sketch.} We use $\length$ hash functions, each with domain size $\width$ and the total communication cost is $\size = \length\cdot\width$ words\footnote{In our experiments, $\nspu\ns$ will be between $\sim1000$ and $\sim10, 000$, and hence one word is enough  to store an entry in the sketch.}. In the $R$-round setting, for each round $r$, we loop over all $x \in \cX$ and compute an estimate $\hat{h}_r(x)$ and the recovered heavy hitters are those with $\sum_{r}\hat{h}_r(x) \ge \tau$. Note that in the open-domain setting, $d = |\cX|$ can be large and this decoding procedure can be computationally infeasible. There are more computationally feasible variants including tree-based decoding but these come at the cost of higher communication cost or lower utility. We stick to the described version in this work and show that our proposed algorithms outperform this computationally expensive version. The advantage will be be at least as large when comparing to the more computationally feasible versions. %

\paragraph{IBLT-based method.} 
In our experiment, each IBLT data structure is of size $8L_0$, where $L_0$ is the targeted capacity for IBLT.
We state more details about how the size is computed in \cref{sec:iblt_app}.

We consider fixed subsampling and adaptive subsampling. 
For fixed subsampling, when the max number of items in each round is upper bounded by $\mmax$, we set the subsampling parameter in \cref{alg:threshold_sampling}, to be $ t = \max\{1, \min\{ \frac{\mmax}{L_0}, \frac{\thr}{2}\} \}$.  In practice, $\mmax$ can be obtained by system parameters such as the number of users in a round and the maximum contribution by a single user.
Setting $t\le \thr/2$ guarantees that the heavy hitters will be kept with high probability ( \cref{lem:threshold}). And setting $t = \frac{\mmax}{L_0}$ guarantees that with high probability, the decoding of IBLT in each round will succeed and we can obtain more information. We set $b = 1$ in our experiments, the estimated and the heavy hitters are defined as those with estimated frequency at least $\thr$.
In the adaptive algorithm (\cref{alg:adaptive_iblt}), for the update rule, we use
\[
    t_{r + 1} =
    0.5 t_r + 0.5 t_r \times \frac{\hat{s}_r}{C}\]
We leave designing better update rules as future work.  

\paragraph{Client data simulation.} We take the ground-truth distribution of strings in the Stackoverflow dataset in Tensorflow Federated and truncate them to the first $3$ characters in set $\rm ROOT$. This is to make sure that the computation is feasible for Count-median Sketch. And the data universe is of size $d = 97336$. In each round, we take $M_r \sim \cN(M, M/10)$ \iid~samples from the this distribution and encode them using the algorithms mentioned above. In the experiment, we assume all samples come from different users ($\nspu = 1$). For Count-sketch, this won't affect the performance. For IBLT with threshold sampling, this will be equivalent to IBLT with binomial sampling. And by \cite{Duffield2005threshold}, this will only increase the variance of the noise introduced in the sampling process. %
The metric we use is the F1 score of real heavy hitters (heavy hitters with true cumulative frequency at least $\thr$) and the estimated heavy hitters.

We take $R \in \{10, 30, 50, 100\}$, $\thr \in  \{20, 50,$ $100, 200\}$, $M \in \{1000, 2000, 5000, 10000 \}$ and $C \in \{100, 200, 500, 1000, 2000, 5000, 8000, 10000, 20000,$ $30000, 40000\}$. For Count-median method, we take the max F1 score over all $H \in \{5, 7, 9, 11\}$ for each communication cost. We run each experiment for 5 times and compute the mean and standard deviation of the obtained F1 score. Our proposed algorithms consistently outperforms the sketching based method. Below we list and analyze a few representative plots. 

In \cref{fig:communication_m10000}, we plot the F1 score comparison under different communication costs when $R = 30, \thr = 50, M = 10000$. It can be seen that our proposed algorithms significantly outperforms the Count-sketch method. Among the IBLT-based methods, Subsampled IBLT with adaptive tuning is performing the best. For non-interactive algorithms, subsampled IBLT with fixed subsampling probability is better compared to the unsampled counter part for a wide range of capacity. 

\begin{figure}[h]
\centering
\includegraphics[width = 0.45\textwidth]{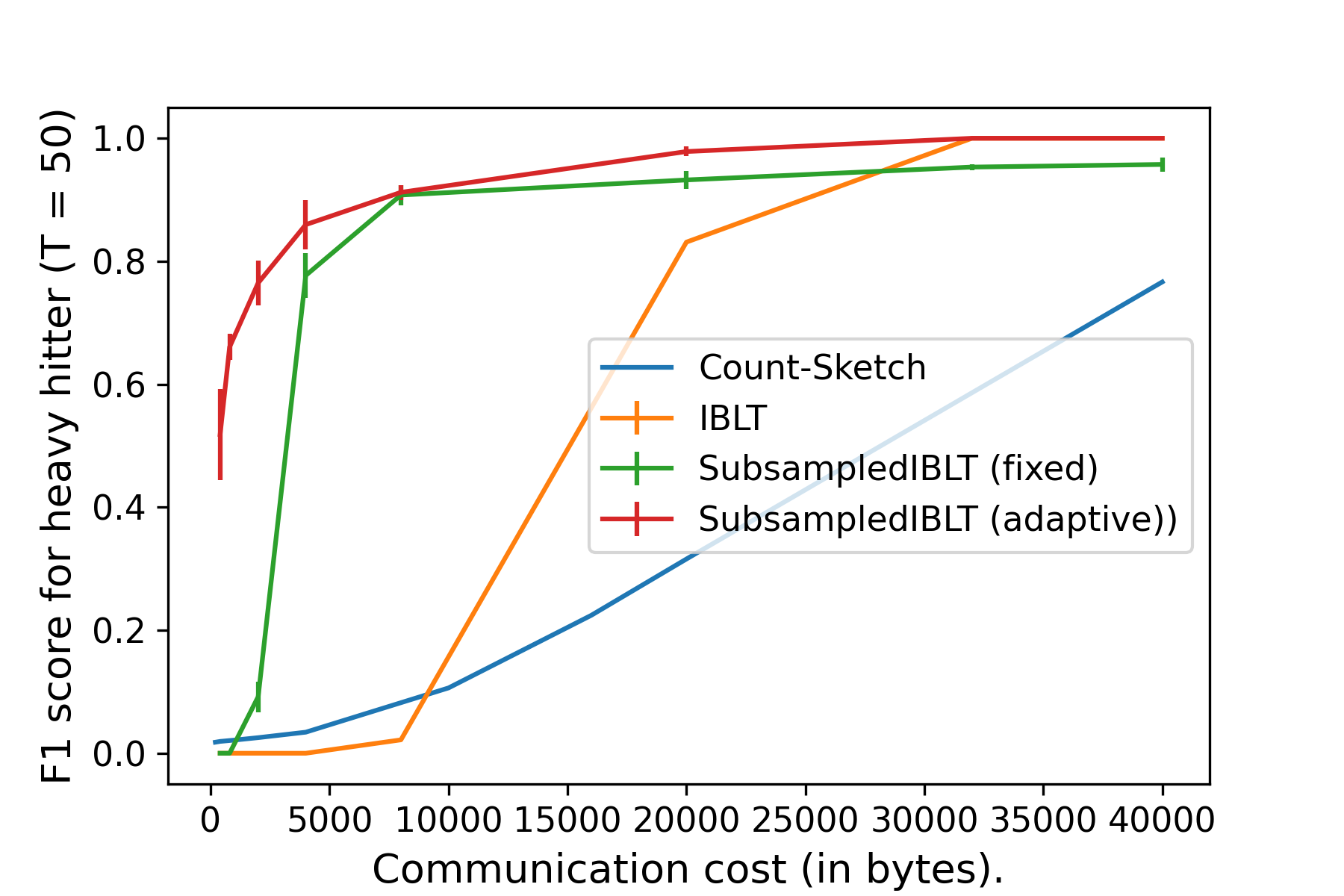}
\caption{F1 score comparison under different communication cost ($R = 30, \thr = 50, M = 10000$). Each F1 score is an average of 5 runs and the error bar represents 3x the standard deviation.}
\label{fig:communication_m10000}
\end{figure}

In \cref{fig:round_m10000}, we plot the F1 score comparison under different round numbers when $C = 10000, \thr = 200, M = 10000$. As we can see, the performance of Count-sketch decreases significantly when the number of rounds increase while the performance of IBLT-based methods remains relatively flat, which is consistent with the theoretical results\footnote{The communication complexity of SubsampledIBLT is $\tilde{\Theta}(mN/\tau R) =\tilde{\Theta}(mn/\tau)$, which depends on $R$ at most logarithmically when $\ns$ is fixed}. The slight increase in the F1 score when $R$ increases might be due to the \iid~generating process of the data in each round. As $R$ increases, we get more information about the underlying distribution and this effect outweighs additional noise introduced by multiple rounds. Better understanding of this effect is an interesting further direction.

\begin{figure}[h]
\centering
\includegraphics[width = 0.45\textwidth]{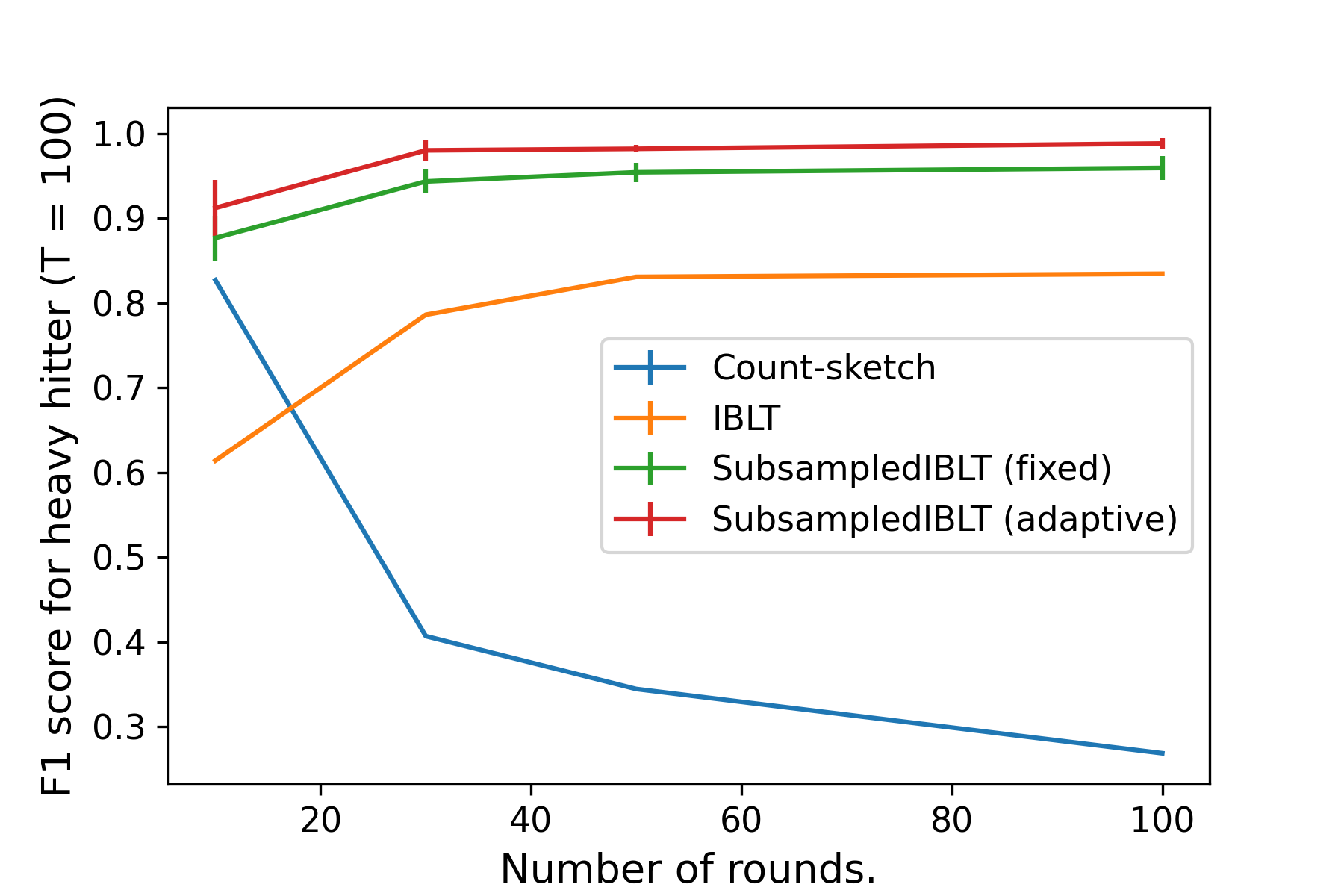}
\caption{F1 score comparison with different number of rounds ($\thr = 100, M = 10000, C = 10000$). Each F1 score is an average of 5 runs and the error bar represents 3x the standard deviation of the runs. }
\label{fig:round_m10000}
\end{figure}

In \cref{fig:adaptive_m10000}, we further demonstrate our adaptive tuning method by showing that it is comparable with the best possible subsampling parameter in a candidate set. More specifically, we run subsampled IBLT with $t \in \{  1,   1.25,   2  ,   5  ,  10  ,  20  ,  50  , 100   \}$ for all communication costs. And the F1 score for SubsampledIBLT (best fixed) is defined as the best F1 score among these candidates. Our result shows that the performance of tha adaptive algorithm is in-par with the best fixed subsampling parameter. It outperforms the best fixed subsampling parameter in certain cases because the set of subsampling parameters we choose from has limited granularity and hence the adaptive algorithm might find better parameters for the underlying instance.

\begin{figure}[h]
\centering
\includegraphics[width = 0.45\textwidth]{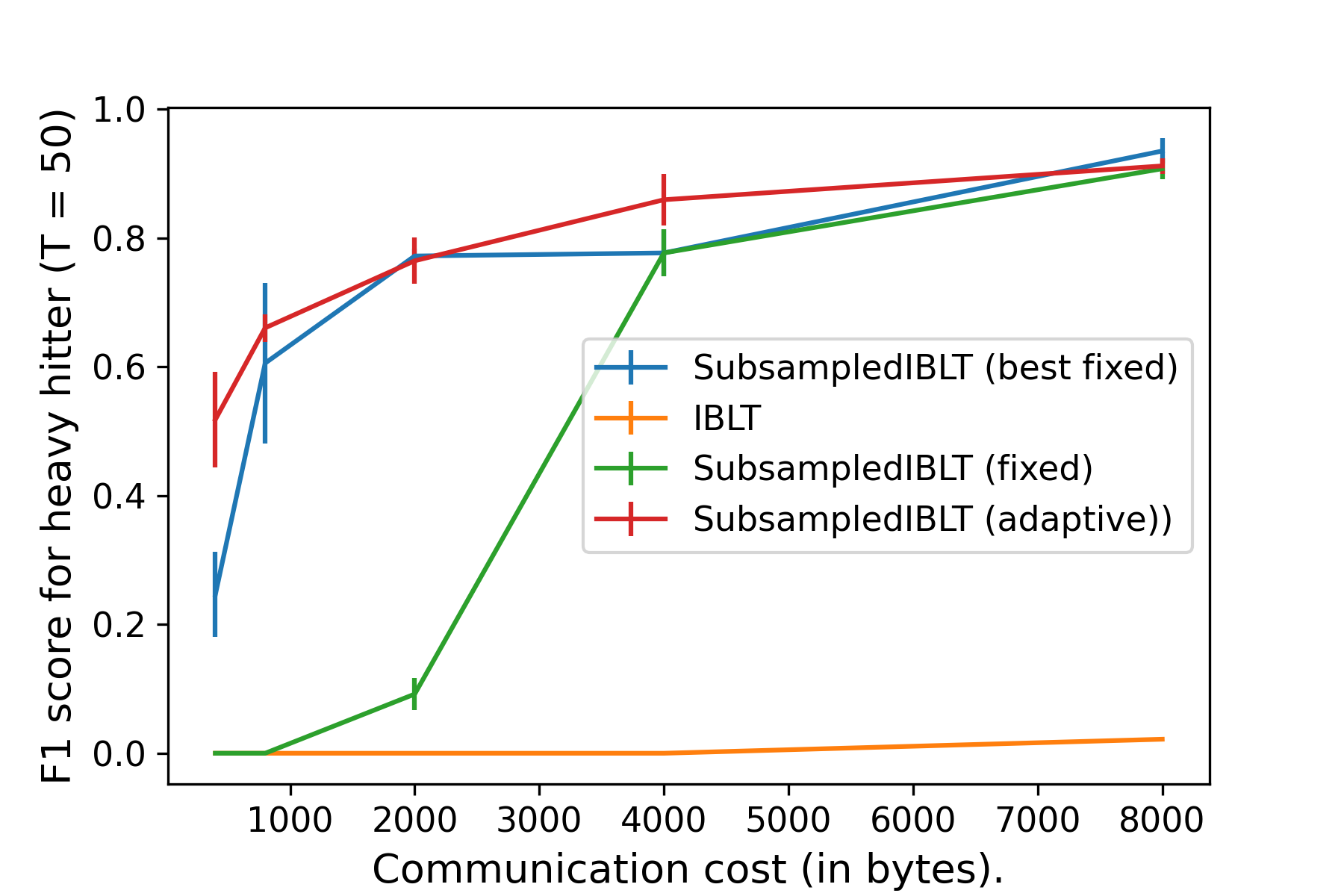}
\caption{F1 score comparison (adaptive vs best fixed probability) ($\thr = 50, M = 10000, C = 5000$). Each F1 score is an average of 5 runs and the error bar represents 3x the standard deviation of the runs.}
\label{fig:adaptive_m10000}
\end{figure}
\section{Conclusion}
We provided lower bounds and matching upper bounds for central tasks in multi-round distributed data analysis: heavy hitters recovery and approximate histograms over large domains. Our findings show how porting single-round approaches based on standard sketching does not achieve optimality, and how this can be cleverly achieved via subsampled IBLTs. Several interesting and non-trivial questions remain to be addressed, including (a) developing distributed differential privacy schemes that are provably optimal for this problem, and (b) developing (non-linear) cryptographic (or other secure) primitives that allow us to extract heavy hitters with smaller (sublinear in $mn$) communication.

\new{
\section{Acknowledgments}
The authors thank Wennan Zhu for early discussions on the work, and 
Badih Ghazi, Ravi Kumar, Pasin Manurangsi, Rasmus Pagh, Amer Sinha, and Ameya Velingker for proposing IBLT as the linear data structure in federated heavy hitter recovery.}

\bibliography{references}
\bibliographystyle{plainnat}

\newpage
\appendix 
\onecolumn
\section{Proof of \cref{thm:ahh}}
\label{sec:proof_ahh}
Note that the algorithm can be viewed as $\rep \eqdef \ceil{20\log(\frac{40\nspu\ns R}{\tau\beta})}$ independent runs of a basic protocol, each of which returns a list $\heavy_i$ of potential heavy hitters. We assume $b \ge 260$, else we take $b' = \max\{b, 260\}$ and the result will change by at most a constant factor.

The next lemma states that the probabilities of heavy elements and tail elements falling in the list.
\begin{lemma} \label{lem:single_run}
    All $\heavy_j$ defined in \cref{alg:subsample_IBLT} satisfy that, if $h^{[R]}(x) \ge \thr$,
    \[
    \probof{x \in \heavy_j} \ge 4/5.
    \]
    Else if $h^{[R]}(x) \le \thr/10$, 
    \[
        \probof{x \in \heavy_j} \le \frac{2h^{[R]}(x)}{\thr}.
    \]
\end{lemma}

Before proving the lemma, we first show how \cref{thm:ahh} can be implied by \cref{lem:single_run}. 

By \cref{lem:single_run}, for $x$ with $h^{[R]}(x) \ge \thr$, we have 
\[
    \probof{x \in \heavy} \ge \probof{{\rm Binom}\Paren{\rep, 4/5} \ge \rep/2} \ge 1 - \frac{\beta \tau}{40 \nspu \ns R},
\]
where the last inequality follows from standard concentration bounds for Binomial random variables  (\eg Chernoff bound \cite{Mitzenmacher2017probability}).

Hence by union bound, we have 
\[
    \probof{\{x \in [d] \mid h^{[R]}(x) \ge \thr \} \subset \heavy } \ge 1 - \frac{\beta}{40}.
\]

For any $x$, with $h^{[R]}(x) \le \thr/10$, by \cref{lem:single_run}, we have
\begin{align*}
\probof{x \in \heavy} \le \probof{{\rm Binom}\Paren{\rep, \frac{2h^{[R]}(x)}{\thr}} \ge \rep/2} \le \frac{\rep + 1}{2} \Paren{\frac{8e}{5} \cdot \frac{2h^{[R]}(x)}{\thr}}^{\rep/2},
\end{align*}
where the last inequality follows from Binomial tail bound (see \cref{lem:binomial_tail}).

Hence by union bound we have 
\begin{align}
    & \qquad \probof{\{x \in [d] \mid h^{[R]}(x) \le \thr/10 \} \cap H \neq \emptyset}  \nonumber \\
    & \le \sum_{x: h^{[R]}(x) \le \thr/10} \frac{\rep+1}{2} \Paren{ \frac{16e h^{[R]}(x)}{5\thr}}^{\rep/2}  \nonumber  \\
    & \le \frac{20\nspu \ns R}{\thr} \frac{\rep+1}{2} \Paren{\frac{8e}{25}}^{\rep/2} \label{eqn:combine_elements}\\
    & \le  \frac{20\nspu \ns R}{\thr} e^{-\frac{b}{20}} \label{eqn:algebra}\\
    & \le \frac{\beta}{2}, \nonumber
\end{align}
where \eqref{eqn:combine_elements} follows from $x^{\rep/2} + y^{\rep/2} \le (x + y)^{\rep/2}$, and hence we can combine symbols to increase the sum of tail probability and end up with at most $\frac{20\nspu \ns R}{\thr}$ symbols with frequencies at most $\thr/10$. \eqref{eqn:algebra} follows from the inequality $(x + 1/2) (8e/25)^x \le e^{-x/10}$ for $x \ge 130$.

By union bound, we get the guarantee claimed in \cref{thm:ahh}.

\begin{proofof}{\cref{lem:single_run}}
The proof mainly consists of two parts. We will first show that local subsampling will keep each heavy hitter with a high probability and each tail element with a low probability, stated in \cref{lem:threshold}. We will then show that after local subsampling, the number of unique elements in each round will decrease so that the decoding in \cref{alg:subsample_IBLT} will succeed with high probability.

\begin{lemma}\label{lem:threshold}
    Let $h_j^{'[R]}$ be the aggregation of locally subsampled histogram for run $j$, \ie
    \[
        h_j^{'[R]} = \sum_{r \in [R]} \sum_{i \in B_r} h_{i, j}.
    \]
    Then if $h^{[R]}(x) \ge \tau$, 
    \[
        \probof{h_j^{'[R]}(x) > 0 } \ge 1 - \frac{1}{e^2}. 
    \]
    Else if $h^{[R]}(x) \le \tau/10$,
    \[
        \probof{x \in \heavy_j} \le \frac{2h^{[R]}(x)}{\thr}.
    \]
\end{lemma}
\begin{proof}
When $h^{[R]}(x) \ge \tau$,
\begin{align*}
    \probof{h_j^{'[R]}(x) > 0 } = 1 - \Pi_{r \in [R], i \in B_r} \min \{1 -  \frac{2 h_{i, j}(x)}{\thr}, 0\} \ge 1 - \Pi_{r \in [R], i \in B_r} e^{-\frac{2 h_{i, j}(x)}{\thr}} =  1 -e^{-\frac{2 h^{[R]}(x)}{\thr}} \ge 1 - \frac{1}{e^2}.
\end{align*}
When $h^{[R]}(x) \le \tau/10$
\begin{align*}
     \probof{h_j^{'[R]}(x) > 0 } = 1 - \Pi_{r \in [R], i \in B_r} \Paren{1 -  \frac{2 h_{i, j}(x)}{\thr}} \le 1 - \Paren{1 - \sum_{r \in [R], i \in B_r}\frac{2 h_{i, j}(x)}{\thr}} = \frac{2h^{[R]}(x)}{\thr}.
\end{align*}
\end{proof}

The next lemma shows that with high probability, the number of elements in each round will decrease by least a factor of $\tau$. 
\begin{lemma} \label{lem:max_zero}
    With probability at least $1 - 1/32$, we have
    \[
        \max_{r \in [R]}\left\{ \norm{h'_r}_0 \right\} = O \Paren{\frac{\nspu \ns}{\tau} \log R}.
    \]
\end{lemma}
\begin{proof}
Since all rounds are independent, it would be enough to show that $\forall i$, with probability at least $1 - 1/32R$, we have 
\[
    \norm{h'_r}_0  = O \Paren{\frac{\nspu \ns}{\tau} \log R}.
\]
To see this, we have
\[
    \probof{\norm{h'_r}_0 \ge \frac{2\nspu \ns}{\tau} \log R} \le \probof{{\rm Binom}\Paren{mn, \frac1{\tau}}\ge \frac{2\nspu \ns}{\tau} \log R} \le \frac{1}{32R},
\]
where the first step follows from that the left hand side is maximized when all $\nspu\ns$ elements in $h_r$ are distinct, and the second step follows from standard binomial tail bound when $mn > 4 \tau$ and $R > 32$.
\end{proof}

Finally, it would be enough to show that when the condition in \cref{lem:max_zero} holds, the decoding of the aggregated IBLT will succeed with high probability. This is true since by \cref{lem:iblt} and union bound, we have
\[
    \probof{\forall j, \hat{h}_j^{[R]} = h_j^{'[R]}} \ge 1 - R \cdot (\frac{\nspu \ns}{\tau} \log R)^{2-\gamma} \ge 1 - 1/32,
\]
where the last inequality holds when $mn > 4 \tau$ and $R > 32$.
Combining the above and \cref{lem:max_zero,lem:threshold}, we conclude the proof since $1/e^2 + 1/32 + 1/32 \le 1/5$.
\end{proofof}

\section{Proof of \cref{thm:ahist_r}} \label{sec:ahist_app}

\new{We start with the case when $\tau \le \sqrt{R}$. In this case, \cref{alg:ahist_r} implements \cref{alg:subsample_IBLT} with $\tau = 1$ and returns the obtained histogram in Line 11. Notice that when $\tau = 1$, the subsampling step is trivial and each user encodes their entire histogram. Hence as long as long the decoding of IBLT succeeds (as promised in the performance analysis of \cref{alg:subsample_IBLT}), we recover the histogram perfectly, \ie $\hat{h}^{[R]} = h^{[R]}.$ And the communication cost will be $\tilde{\Theta}(mn)$.
}

Next we focus on the case when $\tau \ge \sqrt{R}.$ We will condition on the event that the list $H$ obtained in Line 8 of \cref{alg:ahist_r} is a $\thr$ approximate heavy hitter set and hence setting $\hat{x} = 0$ for $x \notin H$ won't introduce error larger than $\thr$.

The rest of the proof follows similarly as the standard proof for Count-sketch. Since $\rep = \ceil{\log\Paren{\frac{4\nspu\ns R}{\tau \beta}}}$, it would be enough to prove that $\forall x \in \cX$, with probability at least 2/3, we have
\[
    |\sum_{r \in [R]} T_r(j, \hash_{j}(x)) \cdot s_{j, r}(x)  - h^{[R]}(x)|  = O(\thr).
\]
Let
\begin{align*}
       \hat{h}_j(x) \eqdef & \sum_{r \in [R]} T^{(r)}(j, \hash_{j}(x))  \cdot s_{j, r}(x) \\ =  & \sum_{r \in [R]}  \sum_{x'}\indic{ \hash_{j}(x') = \hash_{j}(x)} s_{j, r}(x') s_{j, r}(x) \cdot h^{(r)}(x') \\
        = & \sum_{x'}\indic{ \hash_{j}(x') = \hash_{j}(x)} \sum_{r \in [R]} s_{j, r}(x') s_{j, r}(x) \cdot h^{(r)}(x') \\
        = & h^{[R]}(x) +  \sum_{x'\neq x}\indic{ \hash_{j}(x') = \hash_{j}(x)} \sum_{r \in [R]} s_{j, r}(x') s_{j, r}(x) \cdot h^{(r)}(x') 
\end{align*}

Then we have $\expectation{ \hat{h}_j(x) = h^{[R]}(x)}$. Next we provide a bound on the variance.  
Let $H_{10\tau/\sqrt{R}}$ be the set of elements with frequency at least $10\tau/\sqrt{R}$, then we have $|H_{10\tau/\sqrt{R}}| \le \frac{mn\sqrt{R}}{10\tau}$. Since $w = \ceil{\frac{10mn\sqrt{R}}{\tau}}$, we have with probability at least 5/6, 
\[
    \sum_{x' \in H_{10\tau/\sqrt{R}}, x' \neq x}\indic{ \hash_{j}(x') = \hash_{j}(x)}  = 0.
\]
Conditioned on this event, we have
\begin{align*}
    \expectation{\Paren{\hat{h}_j(x)  - h^{[R]}(x)}^2}  & = \expectation{\Paren{\sum_{x' \notin H_{10\tau/\sqrt{R}}, x' \neq x}\indic{ \hash_{j}(x') = \hash_{j}(x)} \sum_{r \in [R]} s_{j, r}(x') s_{j, r}(x) \cdot h^{(r)}(x')}^2}
    \\
    & \le \frac{\max_{x' \notin H_{10\tau/\sqrt{R}}} h^{[R]}(x) \sum_{x' \notin H_{10\tau/\sqrt{R}}} h^{[R]}(x) }{w} \\
    & \le \thr^2.
\end{align*}

Hence with probability at least $5/6$, we have
\[
     \expectation{\absv{\hat{h}_j(x)  - h^{[R]}(x)}} \le \sqrt{6}\tau.
\]
We conclude the proof by a union bound over the two events.
\section{Additional details on IBLT}
\label{sec:iblt_app}
\paragraph{Intuition on \texttt{ListEntries} for IBLT.} The intuition behind the IBLT construction is as follows:
Start with an array $\mathcal{B}$
of length $\ell$ containing 4-tuples
of the form $(0,0,0,0)$.
To insert pair $(x, v)$
hash the tuple ($x$, $\tilde{x}$, $v$, $1$) into $k$ locations $l_1, \ldots, l_k$ in $\mathcal{B}$ based on the key $x$, where $\tilde{x} := G(x)$ is a 
hash of $x$ into a sufficiently large domain so that collision probability is sufficiently unlikely. %
Then add, using component-wise sum, ($x$, $\tilde{x}$, $v$, $1$) to the contents of 
$\mathcal{B}$ in all locations $l_1, \ldots, l_k$.
The $\texttt{ListEntries}/\deciblt$ operation corresponds to the result of the following procedure: (1) find an
entry ($x_{sum}$, $\tilde{x}_{sum}$, $v_{sum}$, $j$) such that $G(x_{sum}/j) = \tilde{x}_{sum}/j$ holds, 
(2) add $(x_{sum}/j, v_{sum})$ to the output, and
(3) remove the pair $(x_{sum}/j, v_{sum})$
by subtracting ($x_{sum}$, $\tilde{x}_{sum}$, $v_{sum}$, $j$)
from the entries $l_1', \ldots, l_k'$ in the array $\mathcal{B}$ to which an insertion would add the tuple for key $x_{sum}/j$ and get back to step (1).
The process of listing entries a.k.a ``peeling off" $\mathcal{B}$.
might terminate before the IBLT is empty.
This is the failure procedure in 
Lemma~\ref{lem:iblt}, which corresponds to the natural procedure to find a 2-core in a random graph \cite{Goodrich2011iblt}.

\paragraph{Sketch size.} 
The above intuition corresponds to the IBLT construction variant from~\cite{Goodrich2011iblt}
that can handle duplicates.
It can be implemented with four length $\ell$ vectors with entries in 
$[d], \texttt{Im}(G), [mn], [mn]$,
respectively. 
In terms of concrete parameters (see ~\cite{Goodrich2011iblt} for details), $k = 3, \ell > 1.3 L_0$,
and $G = \mathbb{Z}_p$ with $p = 2^{31}-1$ give good performance, and require
$1.3L_0(32 + \log_2 d + 2\log_2(mn))$ bits. For the experiment setting considered in \cref{sec:exp}, this is will take at most $8L_0$ words.

\paragraph{Cardinality estimation from saturated IBLT.} 
\cref{lem:iblt} tells us that 
a tight bound $L_0$ on the number of distinct non-zero indices in the intended histogram, 
can save us space in an IBLT encoding. %
However, getting that bound wrong results in an undecodable IBLT. 
While in the single round case all is lost, in the multi-round setting
we leverage a property of undecodable IBLTs that helps update our $L_0$ bound for subsequent rounds
after a failed round. This is the main ingredient for our adaptive tuning heuristic presented in Section~\ref{sec:adaptive}.

Let $\mathcal{B}$ be an undecodable IBLT, and let $S$ be the size of the undecoded graph of $\mathcal{B}$. Also let $\ell$ be the size of $\mathcal{B}$, and let $N$ the (unknown) 
number of distinct elements inserted in $\mathcal{B}$ 
(note that $N$ corresponds to the correct bound $L_0$ that enables decoding).
By \cite{molloy2005cores}, we have the following relation: 
For large enough $N$, if $S < \ell$, we have 
\begin{equation}\label{eq:cardinality-from-core}
    \frac{S}{C} = 1 - e^{-x}(1+x) + o(1),
\end{equation}
where $x$ is the greatest solution to 
\begin{equation}\label{eq:cardinality-from-core-2}
    \frac{6N}{C} = \frac{2x}{(1 - e^{-x})^2}.
\end{equation}
Hence we can have an estimate for $N$ (and thus a correct choice for $L_0$ in a subsequent round) 
based on  $S$ and $C$. We first solve \eqref{eq:cardinality-from-core} ignoring the $o(1)$ term to get $x$ and then plug $x$ and $C$ into \eqref{eq:cardinality-from-core-2} to get an estimate for $N$.
As mentioned above we leverage this fact in Section~\ref{sec:adaptive}.

\section{Proof of lower bounds.}

\subsection{Proof of \cref{thm:ahh_lower}}
\label{sec:ahh_lower}
We will focus on the case when $R = 1$ since the claimed bound doesn't depend on $R$ and we can assume there is no data in other $R - 1$ rounds. We will consider the case when $10< \tau < n/4$.

\new{We prove the theorem by a reduction to the set disjointness problem \citep{BARYOSSEF2004702, Jayram09}. The set disjointness problem ($\textsc{Dist}_{t, d}$) considers the setting where $t$ users where user $i$ has a set of elements $S_i \subset \{1, 2, \ldots, d\}$. The goal is to distinguish between the following two chase with success probability at least $4/5$.
\begin{enumerate}
    \item All $S_i$'s disjoint.
    \item There exists $x \in [d]$ such that for all $i, j \in [t]$, $S_i \cap S_j = \{x\}$.
\end{enumerate}
And the goal is to minimize the size of the transcript of all communications among all users. More specifically, we will use the following lemma:
\begin{lemma}[\citep{Jayram09}]\label{lem:set_disjoint}
Any protocol that solves $\textsc{Dist}_{t, d}$ must have a transcript of size at least $\Theta(d/t)$.
\end{lemma}

Next we show that $\textsc{Dist}_{t, d}$ with $t = \tau$ and $d = mn/2$ can be reduced to the approximate heavy hitter problem. We divide users into $\tau + 1$ groups. For $i \in [\tau]$, the $i$th group has $n_i = \ceil{|S_i|/m}$ users. And let $\tilde{S}_i$ be set of all elements held by users in group $i$. We  partition $S_i$ to subsets of size at most $m$ and distribute them to users in group $i$ arbitrarily. This can be done since $mn_i \ge |S_i|$. The total number of users in the first $\tau$ groups is $\sum_{i \in [\tau]}n_i \le \frac{d + \tau}{m} + \tau \le n$. 
The $\tau + 1$ group has $n - \sum_{i \in [\tau]}n_i$ users and each user has zero element. 

Suppose there exists a $\tau$-\ahh~ linear sketch algorithm with communication cost per-user $o(\frac{mn}{\tau})$. When $S_i$'s are disjoint, all elements in $[d]$ will have frequency $1 < \tau/10$. The algorithm should output an empty list. When $S_i$'s have an unique intersection, the element will have frequency $\tau$, and hence the algorithm should output a list with size 1. By distinguishing between the two cases, the \ahh~algorithm can be used to solve $\textsc{Dist}_{t, d}$.

Moreover, under linear sketching constraint, the size of the transcript is the same as the per-user communication. Hence we conclude the proof by noticing that this violates \cref{lem:set_disjoint}.}

\subsection{Proof of \cref{thm:ahist_lower}}

Here we prove a stronger version of the lower bound where in each round $r$, the communication among users is not limited but
the users in $B_r$ must compress $h^{(r)}$ to an element $Y^{(r)} \in G_r$ with $|G_r| \le 2^\ell$, which is observed by the server. And the server will then obtain an approximate histogram $\widehat{h}^{[R]}$ based on $\Pi = (Y^{(1)}, \ldots, Y^{(R)}, U)$. \new{For a given $\dist$, next we show that any protocol with $\ell = o\Paren{ \min\{ \frac{\nspu \ns\sqrt{R}}{\thr}, \nspu \ns \} }$ won't solve $\dist$-approximate heavy hitter with error probability at most 1/5. We will focus on the case when $\tau \ge \sqrt{R}$ and $\ell = o\Paren{\frac{\nspu \ns\sqrt{R}}{\thr}}$. When $\tau < \sqrt{R}$, the bound follows by setting $\tau = \sqrt{R}$ and the fact that the problem gets harder as $\tau$ decreases.} To simply the proof, we assume $R \ge 400$ without loss of generality.

We consider histograms $h^{(r)}, \forall r \in [R]$ supported over the domain $10\ell$ and are generated \iid~from a distribution $P$.
Let $Z$ be uniformly distributed over $\{\pm 1\}^{5\ell}$, and under distribution $P_Z$, we have $\forall r \in [R], i \in [5\ell]$,
\[
    h^{(r)}(2i) = \begin{cases}
    \frac{mn}{5\ell} & \qquad \text{ with prob } \frac{1}{2} + \frac{10}{\sqrt{R}} Z_i. \\
    0 & \qquad \text{ with prob } \frac{1}{2} - \frac{10}{\sqrt{R}} Z_i.
    \end{cases}
\]
and $$h^{(r)}(2i - 1) = 1 - h^{(r)}(2i).$$
It can be check that $\norm{h^{(r)}}_1 = \nspu \ns$ with probability 1. 
We prove the theorem by contradiction. If the protocol solves $\dist$-approximate heavy hitter with error probability at most 1/5, let
\[
    \hat{Z}_i = \indic{\hat{h}^{[R]}(2i) > \frac{mnR}{10\ell}}.
\]

We have
\begin{align*}
    \probof{\hat{Z}_i \neq Z_i} & \le   \probof{{\absv{\hat{h}^{[R]}(2i) - h^{[R]}(2i)}}  \ge \frac{\nspu \ns \sqrt{R}}{\ell}}  + \probof{\absv{ h^{[R]}(2i) - \frac{\nspu\ns R}{5\ell} \Paren{\frac12 + \frac{10}{\sqrt{R}}Z_i}} \ge \frac{\nspu \ns \sqrt{R}}{\ell}} \\
    & \le \frac{1}{5} + \frac{1}{25} = \frac{6}{25},
\end{align*}

where the first probability is bounded by the success probability of the algorithm and the second probability is bounded using Hoeffding bound. Hence we have
\[
    \sum_{i \in [5\ell]} I(Z_i; \Pi) \ge  \sum_{i \in [5\ell]} I(Z_i; \hat{Z}_i) \ge \sum_{i \in [5\ell]} (1 - H(\frac{5}{26})) \ge 2 \ell,
\]
where $H(p)$ is the Shannon entropy of a Bernoulli random variable with success probability $6/25$.

\new{To upper bound $\sum_{i \in [5\ell]} I(Z_i; \Pi)$, we notice that the vector \[\Paren{\frac{5\ell h^{(r)}(2)}{mn}, \frac{5\ell h^{(r)}(4)}{mn}, \ldots, \frac{5\ell h^{(r)}(10\ell)}{mn}}.\]
follows a product distribution with the marginal of each coordinate being a Bernoulli distribution. Hence by standard arguments on communication-limited estimation of product of Bernoulli random variables (\eg in \cite{braverman2016communication, han2021geometric, acharya2020unified}). In particular, following almost the same steps as in \citet[Section 7.1]{acharya2020unified}, }
\[
     \sum_{i \in [5\ell]} I(Z_i; \Pi) \le R \cdot \Paren{\frac{1}{\sqrt{R}}}^2 \ell = \ell,
\]
which leads to a contradiction. This concludes the proof.

\section{Binomial tail bound.}
\begin{lemma} \label{lem:binomial_tail}
Let $X \sim {\rm Binom}(n, p)$ be a binomial distribution, when $n > 10$ and $p < 1/5$, we have
\[
    \probof{X \ge n/2} \le \frac{n + 1}{2} \Paren{\frac{8ep}{5}  }^{n/2}.
\]
\end{lemma}
\begin{proof}
    \begin{align}
        \probof{X \ge n/2} & = \sum_{i = \floor{(\ns+1)/2}}^\ns  \probof{X = i} \nonumber  \\ 
        & = \sum_{i = \floor{(\ns+1)/2}}^\ns  {n \choose i} \Paren{1 - p}^{n - i} p^i \nonumber  \\
        & \le \frac{n + 1}{2} {n \choose \ceil{n/2}}  \Paren{(1 - p)p}^{n/2}  \label{eqn:n_2}\\
        & \le  \frac{n + 1}{2} (2e)^{n/2} \Paren{\frac{4p}5}^{n/2} \label{eqn:stirling}\\
        & = \frac{n + 1}{2} \Paren{\frac{8ep}{5}  }^{n/2}, \nonumber 
    \end{align}
    where \eqref{eqn:n_2} follows from $\probof{X = i}$ is monotonically decreasing when $i \ge n/2$ and \eqref{eqn:stirling} follows from standard bounds on binomial coefficients.
\end{proof}
\end{document}